\documentclass[conference,10pt]{IEEEtran}
\IEEEoverridecommandlockouts
\usepackage{amsfonts, amssymb, amsbsy, amsmath, mathrsfs, amsthm}
\usepackage{algorithmic}
\usepackage{graphicx}
\usepackage{textcomp}
\usepackage{xcolor}
\def\BibTeX{{\rm B\kern-.05em{\sc i\kern-.025em b}\kern-.08em
    T\kern-.1667em\lower.7ex\hbox{E}\kern-.125emX}}

\usepackage{caption} 
\usepackage{subcaption}
\usepackage{gensymb}
\usepackage{multirow}
\usepackage{hyperref, xcolor}
\usepackage{xcolor}
\usepackage[export]{adjustbox}
\usepackage{array}
\usepackage{booktabs}
\usepackage{siunitx}
\usepackage{verbatim}
\usepackage{arydshln} 
\usepackage[]{cite}
\usepackage[capitalize]{cleveref}

\newtheorem{theorem}{Theorem}

\newtheorem{proposition}{Proposition}

\theoremstyle{definition}
\newtheorem{remark}{Remark}

\usepackage[english]{babel}
\usepackage[super, sort&compress]{natbib}
\setcitestyle{numbers,open={[},close={]}}

\newcommand\blfootnote[1]{%
\hypersetup{hidelinks}
  \begingroup
  \renewcommand\thefootnote{}\footnote{#1}%
  \addtocounter{footnote}{-1}%
  \endgroup
}

\newcommand{\bl}[1]{\textcolor[rgb]{0,0,0}{#1}}

\begin{document}

\title{
Beyond-Diagonal RIS: Adversarial Channels and Optimality of Low-Complexity Architectures
}


\author{\IEEEauthorblockN{Atso Iivanainen, Robin Rajamäki, and Visa Koivunen}
\IEEEauthorblockA{\textit{Department of Information and Communications Engineering}, 
\textit{Aalto University, Finland}
}
}

\maketitle

\begin{abstract}
Beyond-diagonal reconfigurable intelligent surfaces (BD-RISs) have recently gained attention as 
an 
enhancement to conventional RISs. BD-RISs allow optimizing not only the phase, 
but also the amplitude responses 
of their discrete surface elements by introducing adjustable 
inter-element couplings. Various BD-RIS architectures have been proposed to optimally 
trade off between average performance and complexity of the architecture. However, little attention 
has been 
paid 
to worst-case performance. 
This paper characterizes novel sets of adversarial channels for which certain low-complexity BD-RIS architectures have suboptimal performance in terms 
of  
received signal power at an intended communications user. 
Specifically, we consider two recent BD-RIS models: the so-called group-connected and 
tree-connected architecture. The derived adversarial channel sets reveal new surprising connections between the two architectures. 
We validate our analytical results numerically, 
demonstrating that adversarial channels can cause a significant performance loss. 
Our results pave the way towards efficient BD-RIS designs that are robust to adversarial propagation conditions and malicious attacks.

\blfootnote{This work was supported in part by projects Business Finland 6G-ISAC, and Research Council of Finland FUN-ISAC (359094).
}
\end{abstract}


\section{Introduction}

Reconfigurable intelligent surfaces (RIS) are 
an emerging topic in 5G and 6G communications \cite{Chepuri_RIS, Bjornson2022ReconfigurableApplications}. 
The idea is to modify the propagation environment to be more favorable for communications by adding programmable reflective elements to it. 
The impedance of these surfaces can be controlled to high spatial granularity, enabling RISs to appropriately redirect or beamform impinging signals. 
RISs are typically passive and do not provide any additional signal amplification \cite{Bjornson2022ReconfigurableApplications}. 
In a conventional RIS architecture, the surface consists of an array of \emph{independently} adjustable phase shifters. 
So-called ``beyond-diagonal'' RISs (BD-RISs) \cite{nerini_global} 
also allow amplitude (in addition to phase) control by introducing adjustable \emph{coupling} between the RIS elements. However, these new BD-RIS architectures can potentially increase the complexity, and thereby cost and power losses of the RIS, as additional controllable impedances between 
elements are introduced  \cite{mutual_coupling}. Hence, a key challenge is designing BD-RIS architectures striking an optimal balance between performance and complexity.

In principle, each BD-RIS element could be connected to all the others. When the number of RIS elements is $N$, this ``fully-connected'' BD-RIS architecture would therefore have on the order of $N^2$ connections. Recently, the so-called tree-connected (TC) 
architecture was proposed \cite{BDRIS_graphtheory} to reduce the number of connections to the order of $N$, while improving performance compared to conventional (single-connected) RISs 
under 
typical propagation conditions. 
In particular, the TC architecture was shown to achieve the 
performance 
of the fully-connected architecture over a wide range of stochastic channel models 
\cite{BDRIS_graphtheory}. 
However, this still leaves open the possibility of \emph{adversarial} channels for which low-complexity BD-RIS architectures could experience significant performance degradation. 
Investigating such cases is especially important 
when the propagation environment is susceptible to 
malicious 
attacks 
aiming to disrupt the intended use of the RIS \cite{malris_wang, malrisbjörnson}, 
\bl{or coarse quantization is employed at the RIS, which may lead to a nonzero probability of an adversarial channel estimate}. 
Hence, studying adversarial channels can help 
design systems 
that are robust to attacks and adversarial propagation conditions.


This paper characterizes novel adversarial channels in case of two BD-RIS architectures that have recently received significant attention in the literature: the group-connected (GC) and tridiagonal TC architectures \cite{Li2023BeyondArchitectures, nerini_global, mutual_coupling,BDRIS_graphtheory}. We show that when the channels belong to certain unfavorable (and uncountably infinite) sets, the respective BD-RIS architectures fall short of achieving optimal performance in terms of maximizing the received power at an intended communications user. 
Our results also reveal a surprising connection between the GC and tridiagonal TC architectures: adversarial channels for the latter comprise of adversarial channels for diverse GC architectures. Our numerical results indicate that the received signal power may be decreased significantly for such adversarial channels. 


\section{Background}
\label{sec: back}

\subsection{System model and problem formulation}
\label{sec: sys-model}

We consider a 
system model consisting of an RIS, and a single-antenna transmitter (Tx) and receiver (Rx). In a typical downlink communications scenario the Tx is a base station and the Rx a 
user device. The received signal is modeld as:
\begin{equation}
    y(t) = hx(t) + n(t),
\end{equation}
where $x(t) \in \mathbb{C}$ is the transmitted signal, $h\in\mathbb{C}$ the channel, and $n(t)$ receiver noise. 
Given an RIS with $N$ unit cells, we can model the scalar channel between the Tx and the Rx as: 

\begin{equation}
    h = h_d + \mathbf{h}_R^H \mathbf{\Theta}\mathbf{h}_T, 
\end{equation}

\noindent where $\mathbf{h}_T \in \mathbb{C}^N$ is the vector channel between the Tx and the RIS, $\mathbf{\Theta} \in \mathbb{C}^{N \times N}$ is a so-called "scattering matrix" describing the effect of the RIS on the signal, and $\mathbf{h}_R \in \mathbb{C}^N$ is the channel between the RIS and the Rx. The direct channel between the Tx and the Rx is denoted by $h_d \in \mathbb{C}$. The RIS is typically a planar surface with $N$ independently adjustable elements, each of which can perform a phase shift to the signal. 
A BD-RIS can also change the element amplitudes in addition to phase shifts. 
To analyze the RIS performance, we assume that $h_d=0$ which corresponds to the scenarios of high practical relevance where the direct line of sight path is blocked \cite{Bjornson2022ReconfigurableApplications}. 
\bl{The case of a nonzero (but small) $h_d$ is left for future work.}

Now the received signal power $P_R$ at Rx is defined as \cite{nerini_global}: 
\begin{equation}
    \label{eq: obj-function}
    P_R \triangleq P_T |\mathbf{h}_R^H \mathbf{\Theta} \mathbf{h}_T|^2,
\end{equation}
\noindent where $P_T$ is the transmitted signal power. Without loss of generality, we henceforth assume $P_T = 1$ to simplify the notation. Assuming a passive and lossless RIS, $\mathbf{\Theta}$ is a symmetric unitary matrix defined as \cite{micro_eng, nerini_global, graph_theory}: 

\begin{equation}
    \mathbf{\Theta} \triangleq (\mathbf{I} + jZ_0 \mathbf{B})^{-1}(\mathbf{I} - jZ_0\mathbf{B}), 
    \label{eq: theta}
\end{equation}
where $\mathbf{B}\in \mathbb{R}^{N \times N}$ is the so-called susceptance matrix of the RIS. 
For a passive and lossless network, assumed here, $\mathbf{B}$ is symmetric, i.e., $\mathbf{B} = \mathbf{B}^T$ \cite{micro_eng}.

Our objective is to find $\mathbf{\Theta}$ s.t. the received signal power is maximized \cite{BDRIS_graphtheory} (assuming $P_T = 1$):
\begin{align}
    &\underset{\mathbf{B} \in \mathcal{B} \subseteq \mathbb{R}^{N \times N}}{\text{maximize}} \ 
    |\mathbf{h}_R^H \mathbf{\Theta} \mathbf{h}_T|^2  \label{eq: max} \\
    &\text{s.t.} \ \mathbf{\Theta} = (\mathbf{I} + jZ_0 \mathbf{B})^{-1}(\mathbf{I} - jZ_0\mathbf{B}) \text{ and } \mathbf{B}=\mathbf{B}^T,   \notag 
\end{align}
where $\mathbf{h}_T\in \mathbb{C}^{N}$ and $\mathbf{h}_R \in \mathbb{C}^{N}$ are the channels from Tx to RIS and from RIS to Rx, respectively, assumed to be known. Set $\mathcal{B}$ is a constraint set on optimization variable $\mathbf{B} \in \mathbb{R}^{N \times N}$, determined by the RIS-architecture at hand \cite{mutual_coupling}.

An upper bound on $P_R$, denoted by $\bar{P}_R$, is obtained using the Cauchy-Schwartz inequality and the norm preserving property of unitary transforms:
\begin{align}
    \label{eq: bound_full}
    \bar{P}_R &\triangleq  ||\mathbf{h}_R||_2^2||\mathbf{h}_T||_2^2  \\
    &= ||\mathbf{h}_R||_2^2 ||\mathbf{\Theta} \mathbf{h}_T||_2^2 \geq |\mathbf{h}_R^H \mathbf{\Theta} \mathbf{h}_T|^2 = P_R.  \notag  
\end{align}
Section \ref{sec: new_results} establishes 
for which 
channel vectors $(\mathbf{h}_R,\mathbf{h}_R)$ 
the solution to \eqref{eq: max} achieves \eqref{eq: bound_full} \bl{for various} 
BD-RIS architectures. 

\subsection{Mathematical representations of different RIS-architectures}
\label{subsec: architectures}

The structure of $\mathbf{B}$ depends on the RIS architecture. BD-RIS models are conventionally divided into single, group and fully-connected architectures \cite{mutual_coupling}. Informally, 
a \emph{connected} 
architecture 
is one 
where a path can be found between any two elements \cite{BDRIS_graphtheory}. This paper considers the single, group, and recently proposed tree-connected \cite{BDRIS_graphtheory} architecture.

A \emph{single-connected} architecture is equivalent to a conventional RIS, which is constructed so that each unit cell is controlled through an independent admittance circuit that alters the phase shift of the reflected signal. As there are no connections between the elements, the network is modeled by a diagonal $N \times N$ susceptance matrix $\mathbf{B} = \text{diag}(b_1, b_2,...,b_N)$.

A \emph{fully-connected} architecture implements a single impedance circuit, where each element is directly linked to one another. Consequently, the susceptance matrix only has non-zero entries, i.e, $B_{i,\ell}\neq0,\forall i,\ell$. 

In a \emph{group-connected} architecture, the network consists of mutually disconnected groups, such that 
$\mathbf{B}$ 
is \emph{block-diagonal}: 
\begin{equation}
    \label{eq: block_diagonal}
    \mathbf{B} = \text{diag}(\mathbf{B}_1, \mathbf{B}_2, ... , \mathbf{B}_G). 
\end{equation}
Here, $\mathbf{B}_g=\mathbf{B}_g^T$ is a (symmetric) connected $N_g \times N_g$ submatrix, $G$ is the number of groups, $g$ is the group index, and $N_g$ the number of elements in group $g$.

A \emph{tree-connected} architecture is a connected architecture with the minimum number of $N-1$ connections, such that there is a path between each element, albeit not necessarily a direct one. An illustrative example of this kind of design is a tridiagonal RIS, in which each cell element is connected to its two nearest neighbors, as depicted 
in Figure \ref{fig: tree-connected}. In this case, the susceptance matrix 
is \emph{tridiagonal} (and symmetric):

\begin{equation}
\label{eq: tridiagonal}
\mathbf{B} = \begin{bmatrix}
    B_{1,1} & B_{1,2} & \dots & 0\\
    B_{1,2} & B_{2,2} &  \ddots & \vdots\\
    \vdots  & \ddots & \ddots & B_{N-1,N} \\
    0  & \dots & B_{N-1,N} & B_{N,N} 
\end{bmatrix}.
\end{equation}
There are multiple other tree-connected BD-RIS architectures beyond the tridiagonal case considered in this paper. 
Indeed, it was shown in \cite{BDRIS_graphtheory} that any tree-connected architecture can reach the received power upper bound \eqref{eq: bound_full} with probability $1$ when the entries of channel vectors $(\mathbf{h}_R,\mathbf{h}_T)$ are drawn independently from a continuous probability distribution. 
However, the question remains: do \emph{adversarial} channels exist for which \eqref{eq: bound_full} \emph{cannot} be achieved? In the next section, we answer this question affirmatively, characterizing a set of channels for which the tridiagonal RIS cannot achieve \eqref{eq: bound_full}. Interestingly, these adversarial channels are intimately related to adversarial channels of certain group-connected BD-RISs, as we show. 
\begin{figure}[h]
    \centering
\includegraphics[width=80mm,trim={0cm 0 0 2.5cm},clip]{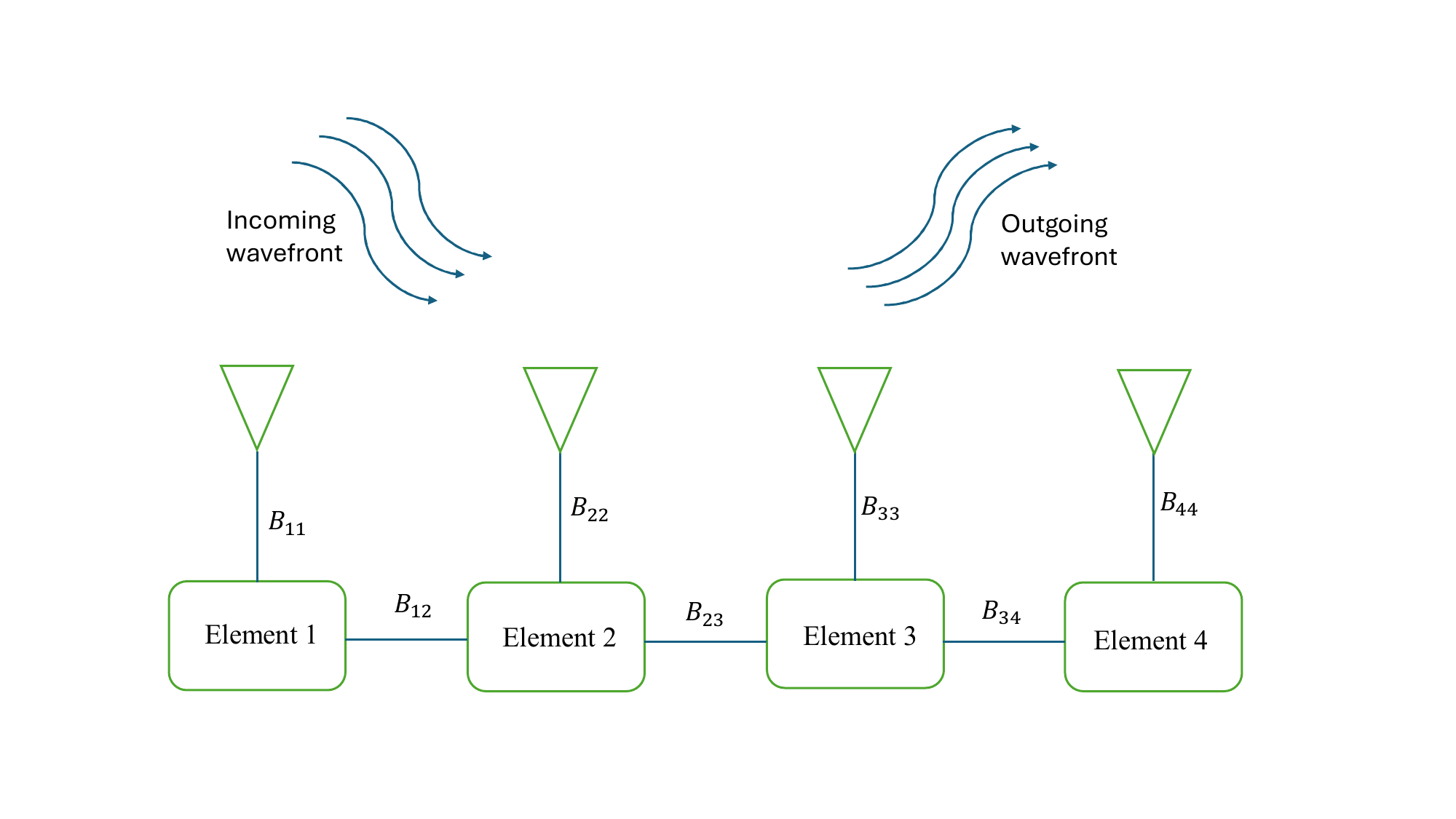}
    \vspace{-.5cm}
    \caption{
    Tridiagonal BD-RIS architecture}
    \label{fig: tree-connected}
\end{figure}

\section{Adversarial channels for BD-RISs}
\label{sec: new_results}

The group and tree-connected architectures have been shown to improve received signal power $P_R$ compared to the conventional single-connected RIS. Indeed, the upper bound $\bar{P}_R$ in \eqref{eq: bound_full} can be achieved for a wide variety of channels \cite{BDRIS_graphtheory, nerini_global}. 
In this section, we characterize for which channels the maximum receiver power $\bar{P}_R$ can, respectively cannot, be achieved with the group-connected and tridiagonal architectures. 

\subsection{Group connected BD-RIS architecture}
\label{sec: disconnected}

The group-connected architecture has a limited ability to achieve $\bar{P}_R$, i.e., 
the optimal solution to 
\eqref{eq: max} can fall short of \eqref{eq: bound_full} 
when $\mathbf{B}$ is of the form in \eqref{eq: block_diagonal}. 
%
We associate a specific group-connected architecture with an index set $\mathcal{I}=\{i_g\}_{g=1}^{G-1}$, which is a fixed subset of $G-1$ positive integers smaller than $N$, such that $1 \leq i_1<i_2< \ldots<i_{G-1}< N$. We refer to this as an \emph{$\mathcal{I}$-group-connected} architecture. We denote the indices of the elements belonging to the $g$th group as $\mathcal{J}_g(\mathcal{I})$, where
\[
    \mathcal{J}_g(\mathcal{I}) 
    \triangleq
    \begin{cases}
        \{1 : i_1\}, & \text{if } g = 1, \\
        \{(i_{g-1}+1) : i_g\}, & \text{if } 1 < g < G, \\
        \{(i_{G-1}+1) : N\}, & \text{if } g = G. 
    \end{cases}
\]
To eliminate trivial permutations, we have assumed that mutually connected BD-RIS elements take on adjacent indices, and that the $G$ groups are in nonincreasing order of cardinality, i.e., $|\mathcal{J}_1(\mathcal{I})|\geq |\mathcal{J}_2(\mathcal{I})|\geq \ldots \geq  |\mathcal{J}_G(\mathcal{I})|$. 
We can now express the received power in \eqref{eq: obj-function} as (assuming $P_T = 1$ as before): 
\begin{equation}
    P_R^{GC} = \Bigg| \sum_{g=1}^G [\mathbf{h}_R]_{\mathcal{J}_g(\mathcal{I})}^H \mathbf{\Theta}_g [\mathbf{h}_T]_{\mathcal{J}_g(\mathcal{I})}\Bigg|^2. 
\end{equation}
Using the Cauchy-Schwartz inequality, we can define a tight upper bound $\bar{P}_R^{GC} \leq \bar{P}_R$ on $P_R^{GC}$, which is achieved when each $\mathbf{B}_g$ is optimized individually: 

\begin{align}
    \label{eq: bound_disconnected}
    \bar{P}_R^{GC} &\triangleq \left ( \sum_{g=1}^G ||[\mathbf{h}_R]_{\mathcal{J}_g(\mathcal{I})}||_2 \ ||[\mathbf{h}_T]_{\mathcal{J}_g(\mathcal{I})}||_2 \right )^2  \geq P_R^{GC}. 
\end{align}
We now provide necessary and sufficient conditions on channel vectors $(\mathbf{h}_R,\mathbf{h}_T)$ for \eqref{eq: bound_disconnected} and \eqref{eq: bound_full} to coincide. To this end, let
\begin{align}
     \label{eq: C1}
     \begin{aligned}
        \mathcal{C}_1(\mathcal{I})\!\triangleq\!\{ (\mathbf{u},\mathbf{v})
        \, | \, \nexists \gamma>0 \text{ s.t. } \|&\mathbf{u}_{\mathcal{J}_g(\mathcal{I})}\|_2\!=\!\gamma \|\mathbf{v}_{\mathcal{J}_g(\mathcal{I})}\|_2\\
        &\text{ for all } g\in\{1:G\}\}.
     \end{aligned}  
\end{align}
Set $\mathcal{C}_1(\mathcal{I})$ denotes the adversarial channel vectors in case of the $\mathcal{I}$-group-connected architecture.

\begin{proposition}
    \label{proposition: disconnected}
    
    Given an $\mathcal{I}$-group-connected BD-RIS architecture 
    where $\mathbf{B}$ follows \eqref{eq: block_diagonal}, denoted $\mathbf{B} \in \mathcal{B}_{GC}(\mathcal{I})$, 
    we have: 
    \begin{align} 
        (\mathbf{h}_R, \mathbf{h}_T) \in C_1(\mathcal{I}) \iff
        \underset{\mathbf{B} \in \mathcal{B}_{GC}(\mathcal{I})}{\max} \ |\mathbf{h}_R^H \mathbf{\Theta} \mathbf{h}_T|^2 < \bar{P}_R .
        \label{eq: GC_condition}
    \end{align}
\end{proposition}

\begin{proof}
By \labelcref{eq: bound_disconnected,eq: bound_full}, $\bar{P}_R^{GC} \leq \bar{P}_R$ can be written as
\begin{align}
    \label{eq: cauchy-schwartz}
    &\left ( \sum_{g=1}^G ||[\mathbf{h}_R]_{\mathcal{J}_g(\mathcal{I})}||_2 \ ||[\mathbf{h}_T]_{\mathcal{J}_g(\mathcal{I})}||_2 \right )^2 \leq \notag \\
    &\left (\sum_{g=1}^G ||[\mathbf{h}_R]_{\mathcal{J}_g(\mathcal{I})}||_2^2 \right ) \left ( \sum_{g=1}^G ||[\mathbf{h}_T]_{\mathcal{J}_g(\mathcal{I})}||_2^2 \right ). 
\end{align}
By the Cauchy-Schwartz inequality, equality holds in \eqref{eq: cauchy-schwartz}, i.e., $\bar{P}_R^{GC} = \bar{P}_R$, if and only if $ \|[\mathbf{h}_R]_{\mathcal{J}_g(\mathcal{I})}\|_2 = \gamma  \|[\mathbf{h}_R]_{\mathcal{J}_g(\mathcal{I})}\|_2, \forall g$ and some $\gamma > 0$, i.e., $(\mathbf{h}_R,\mathbf{h}_T)\notin\mathcal{C}_1(\mathcal{I})$. Hence, $\bar{P}_R^{GC} \neq \bar{P}_R$, i.e., $\bar{P}_R^{GC} < \bar{P}_R$, if and only if $(\mathbf{h}_R,\mathbf{h}_T)\in\mathcal{C}_1(\mathcal{I})$.
%
\end{proof}

\begin{remark}
    \cref{proposition: disconnected} implies that the $\mathcal{I}$-group-connected architecture 
    achieves upper bound $\bar{P}_R$ if and only if
    \begin{align}
        \label{eq: gc_reaches_upper_bound}
        \exists \gamma>0 \text{ s.t. } \|
        [\mathbf{h}_R]_{\mathcal{J}_g(\mathcal{I})}\|_2\!=\!\gamma \|[\mathbf{h}_T]_{\mathcal{J}_g(\mathcal{I})}\|_2, \forall g.
    \end{align}
\end{remark}

\begin{remark}
    If $G = N$, each group only consists of one element, reducing the architecture to a single-connected (conventional) RIS. In this case, $\bar{P}_R$ is achieved if and only if 
    \begin{equation}
        \label{eq: sc_reaches_upper_bound}
        \exists \gamma > 0 \  \text{s.t.} \ |[\mathbf{h}_R]_n| = \gamma |[\mathbf{h}_T]_n|, \forall n.
    \end{equation}
    Most notably, \eqref{eq: sc_reaches_upper_bound} holds for line-of-sight channels where $[\mathbf{h}_R]_n\!=\!c_1 e^{j\phi_n},[\mathbf{h}_T]_n\!=\!c_2 e^{j\varphi_n}$ for some $c_1,c_2,\phi_n,\varphi_n\in\mathbb{R}$.
\end{remark}

\subsection{Tree-connected BD-RIS architecture}

The tree-connected architecture also has adversarial channels, 
which we will now 
characterize. This is the main novel contribution of the paper. 
Our starting point is the following result from \cite{BDRIS_graphtheory}: an optimal $\mathbf{B}$ solving \eqref{eq: max} in case of the tridiagonal architecture 
can equivalently be written as the solution to a certain linear system of equations $\mathbf{A}\mathbf{x} = \mathbf{b}$. Here, $\mathbf{x}\in\mathbb{R}^{2N-1}$ contains the \emph{unknown} entries of 
$\mathbf{B}$,
whereas $\mathbf{A} \in \mathbb{R}^{2N \times (2N - 1)}$ and $\mathbf{b}\in \mathbb{R}^{2N}$ are functions of the \emph{known} channel vectors $\mathbf{h}_R,\mathbf{h}_T$. Specifically,
\begin{align}
    \label{eq: xB}
    \begin{aligned}
    \mathbf{x} &\triangleq \left [ \mathbf{B}_{1,1}, ... , \mathbf{B}_{N,N}, \mathbf{B}_{1,2},...,\mathbf{B}_{N-1,N} \right]^T,
    \end{aligned}\\
    \label{eq: A_matrix}
    \mathbf{A}\!\triangleq\!\left [\begin{array}{cc}
    \text{Re}\{\mathbf{A}_1\} & \text{Re}\{\mathbf{A}_2\}\\
    \hdashline
    \text{Im}\{\mathbf{A}_1\} & \text{Im}\{\mathbf{A}_2\}
    \end{array}
    \right], \mathbf{b}\!\triangleq\!\left [\begin{array}{c}
    \text{Re}\{\hat{\mathbf{h}}_T\!-\!\hat{\mathbf{h}}_R\} \\
    \hdashline
    \text{Im}\{\hat{\mathbf{h}}_T\!-\!\hat{\mathbf{h}}_R\}
    \end{array}\right ],
\end{align}
where 
 $\mathbf{A}_1 \triangleq \text{diag}(\boldsymbol{\alpha}) \in \mathbb{C}^{N \times N}$ is a diagonal matrix; 
 \begin{align}
 \boldsymbol{\alpha} \triangleq jZ_0(\hat{\mathbf{h}}_R + \hat{\mathbf{h}}_T) \in \mathbb{C}^{N} \label{eq: alpha}
 \end{align}
 is an auxiliary vector; $\hat{\mathbf{h}}_R\triangleq \mathbf{h}_R/\|\mathbf{h}_R\|_2$ and $\hat{\mathbf{h}}_T\triangleq \mathbf{h}_T/\|\mathbf{h}_T\|_2$ are the normalized (unit-norm) channel vectors; 
and $\mathbf{A}_2 \in \mathbb{C}^{N \times (N - 1)}$ is a bidiagonal matrix with $(k,l)$th entry
\begin{align}
\label{eq: A2}
    [\mathbf{A}_2]_{k,l} \triangleq \begin{cases}
        \alpha_{k+1} &, \ \text{if} \ k = l \\
        \alpha_{k-1} &, \ \text{if} \ k = l + 1 \\
        0 &, \ \text{otherwise}. 
    \end{cases}
\end{align}
Now, $\mathbf{A}\mathbf{x}=\mathbf{b}$ has a solution if and only if $\mathbf{b}$ is in the range space of $\mathbf{A}$, i.e., $\mathbf{b}\in\mathcal{R}(\mathbf{A})$. However, it is not obvious 
for which 
pairs $(\mathbf{h}_R, \mathbf{h}_T)$ vector $\mathbf{b}\in\mathcal{R}(\mathbf{A})$, 
since $\mathbf{A}$ is a \emph{tall} matrix, and \emph{both} $\mathbf{A}$ and $\mathbf{b}$ are functions of channel vectors $(\mathbf{h}_R, \mathbf{h}_T)$. 
It was argued 
in \cite{BDRIS_graphtheory} that if the channel vectors $(\mathbf{h}_R, \mathbf{h}_T)$ are drawn from a continuous probability distribution, then $\mathbf{b} \in \mathcal{R}(\mathbf{A})$ with probability 1. However, this does not exclude the existence of \emph{adversarial} channels for which $\mathbf{b} \notin \mathcal{R}(\mathbf{A})$. 
To characterize the set of adversarial channels, consider a fixed index set $\mathcal{I}=\{i_g\}_{g=1}^{G-1}$ similarly to \cref{sec: disconnected}. Furthermore, define the set of vector tuples
\begin{align}
     \label{eq: C2}
     \begin{aligned}
    \mathcal{C}_2(\mathcal{I})\!\triangleq\!\{ (\mathbf{u},\mathbf{v})
    \, |\, \exists \gamma_i\!\in\!\mathbb{R} \text{ s.t. }[\hat{\mathbf{u}}+&\hat{\mathbf{v}}]_i\!=\! \gamma_i\cdot[\hat{\mathbf{u}}\!+\!\hat{\mathbf{v}}]_{i+1}\\
    &\text{if and only if } i\in\mathcal{I}\},
    \end{aligned}
\end{align}
where $\hat{\mathbf{u}}\triangleq\mathbf{u}/\|\mathbf{u}\|_2,\hat{\mathbf{v}}\triangleq\mathbf{v}/\|\mathbf{v}\|_2 \in\mathbb{C}^N$. 
 It can be shown that if $(\mathbf{h}_R,\mathbf{h}_T)\in\mathcal{C}_2(\mathcal{I})$, then $\mathbf{A}$ does not have full column rank. However, this alone is insufficient for $\bar{P}_R$ to be unachievable. 
We thus define the set of \emph{adversarial channel vector pairs} as
\begin{align}
    \label{eq: adv_set}
        \begin{aligned}
    \mathcal{A} \triangleq \bigcup_{\emptyset\subset\mathcal{I}\subseteq \{1:N-1\}} \mathcal{C}_1(\mathcal{I})\cap \mathcal{C}_2(\mathcal{I}),
        \end{aligned}
\end{align}
where the union is taken over all $\sum_{k=1}^N\binom{N-1}{k}=2^{N-1}-1$ nonempty subsets of $\mathcal{I} \subseteq \{1,2,\ldots,N-1\}$. 
The following theorem presents our main result---a novel 
sufficient condition characterizing the set of channel vectors for which the tridiagonal BD-RIS architecture cannot achieve 
\eqref{eq: bound_full}.
\begin{theorem}
    \label{theorem: tree-connected}
    Given 
   a tridiagonal BD-RIS architecture where $\mathbf{B}$ follows \eqref{eq: tridiagonal}, denoted $\mathbf{B} \in \mathcal{B}_{TC}$, we have:
    \begin{equation}
        \label{eq: theorem}
         (\mathbf{h}_R, \mathbf{h}_T) \in \mathcal{A} \implies \underset{\mathbf{B} \in \mathcal{B}_{TC}}{\max} \ |\mathbf{h}_R^H \mathbf{\Theta} \mathbf{h}_T|^2 < \bar{P}_R.
    \end{equation}
\end{theorem}
\begin{proof} 
We will use the following property \cite{BDRIS_graphtheory}: $ |\mathbf{h}_R^H \mathbf{\Theta} \mathbf{h}_T|^2 = \bar{P}_R$ if and only if $\mathbf{A}\mathbf{x}=\mathbf{b}$ has a solution, where $\mathbf{A}$, $\mathbf{x}$, and $\mathbf{b}$ are given by \labelcref{eq: A_matrix,eq: xB}. Specifically, we show that if $(\mathbf{h}_R,\mathbf{h}_T)\in\mathcal{A}$, then $\mathbf{A}\mathbf{x}\neq \mathbf{b}$, i.e., $\underset{\mathbf{B} \in \mathcal{B}_{TC}}{\max} \ |\mathbf{h}_R^H \mathbf{\Theta} \mathbf{h}_T|^2 < \bar{P}_R$.

If $(\mathbf{h}_R,\mathbf{h}_T)\in\mathcal{A}$, there exists $\gamma_i\in\mathbb{R}; i\in\{1:N-1\}$ s.t. $\alpha_i = \gamma_i \alpha_{i+1}$, where $ \boldsymbol{\alpha}$ is defined in \eqref{eq: alpha}. 
Rewriting $\boldsymbol{\alpha}=\text{Re}(\boldsymbol{\alpha})+j\text{Im}(\boldsymbol{\alpha})$,
we obtain an expression for $\gamma_i$:
\begin{equation}
    \gamma_i = \begin{cases}
        c_i\in\mathbb{R}\setminus\{0\},& \text{if} \ \alpha_{i+1}=\alpha_{i} = 0\\
        \frac{\text{Re}(\alpha_{i+1})}{\text{Re}(\alpha_{i})},& \text{if} \ \text{Im}(\alpha_{i+1}) = \text{Im}(\alpha_{i}) = 0 \\
        \frac{\text{Im}(\alpha_{i+1})}{\text{Im}(\alpha_{i})},& \text{if} \ \text{Re}(\alpha_{i+1}) = \text{Re}(\alpha_{i}) = 0, \\
        \frac{\text{Re}(\alpha_{i+1})}{\text{Re}(\alpha_{i})} =\frac{\text{Im}(\alpha_{i+1})}{\text{Im}(\alpha_{i})}, &  \text{otherwise}.
    \end{cases}
\end{equation}
Note that the $i$th, ($i+1$)th, and ($N+i$)th columns of $\mathbf{A}$ are
{\footnotesize
\begin{align*}
\mathbf{c}_i&=[0,\ldots,0,\text{Re}(\alpha_i),0,0,\ldots,0,\text{Im}(\alpha_i),0,0\ldots,0]^T,\\
\mathbf{c}_{i+1}&=[0,\ldots,0,0,\text{Re}(\alpha_{i+1}),0,\ldots,0,0,\text{Im}(\alpha_{i+1}),0\ldots,0]^T,\\
\mathbf{c}_{N+i}&=[0,\ldots,0,\text{Re}(\alpha_{i+1}),\text{Re}(\alpha_i),0,\ldots,0,\text{Im}(\alpha_{i+1}),\text{Im}(\alpha_i),0\ldots,0]^T.
\end{align*}
}%
Hence, 
$\mathbf{c}_{N+1}$ 
is a linear combination of 
$\mathbf{c}_i$ and $\mathbf{c}_{i+1}$, 
\begin{equation}
    \mathbf{c}_{N+i} = \gamma_i   \mathbf{c}_i +  \mathbf{c}_{i+1}/\gamma_i ,
\end{equation}
and we can equivalently write
\begin{equation}
    \mathbf{A} \mathbf{x} = \mathbf{A} \mathbf{x}' = \mathbf{b}, 
\end{equation}
where the ($N+i$)th element of $\mathbf{x}'$ is zero, $x_{N+i}' = 0$, and
\begin{align}
\mathbf{x} &= \left [ \begin{array}{ccccccccc}
     x_1, \ x_2, \dots, \ x_i, \ x_{i+1}, \dots, \ x_{N+i}, \dots, \ x_{2N-1}
\end{array}
\right ]^T, \notag \\
\mathbf{x}' &= \left [ \begin{array}{ccccccccc}
     x_1, \ x_2, \dots, \ x_i', \ x_{i+1}', \dots, \ x_{N+i}', \dots, \ x_{2N-1}
\end{array}
\right ]^T, \notag
\end{align}
as well as $x_i' = x_i + \gamma_ix_{N+i}$ and  $x_{i+1}' = x_{i+1} + x_{N+1}/\gamma_i$.

By (\ref{eq: xB}), 
$x_{N+i}$ corresponds to element $B_{i,i+i}=B_{i+1,i}$ of tridiagonal matrix $\mathbf{B}$. Hence, $\mathbf{B}$ can be written as a block diagonal matrix, corresponding to a group connected architecture 
with a missing link between elements $i$ and $i+1$: 
\begin{equation}
    \label{eq: diagonal_reduction}
    \mathbf{B}' = \text{diag}(\mathbf{B}'_1, \mathbf{B}'_2),
\end{equation}
Here, $\mathbf{B}'_1\in\mathbb{R}^{i \times i}$ is a tridiagonal symmetric matrix including the $i$ first diagonal elements and $i-1$ non-diagonal elements of $\mathbf{B}$. Similarly $\mathbf{B}'_2\in\mathbb{R}^{(N-i) \times (N-i)}$ includes the last $N-i$ diagonal elements and $N-i-1$ non-diagonal elements of $\mathbf{B}$.

Now, by assumption, $(\mathbf{h}_R, \mathbf{h}_T) \in \mathcal{A}\subseteq \bigcup_{\mathcal{I}}\mathcal{C}_1(\mathcal{I})$, where 
the union is taken over all possible nonempty $\mathcal{I}\subseteq\{1:N-1\}$. Hence, no group-connected architecture can achieve $\bar{P}_R$, which by the preceding reasoning implies that $\mathbf{A}\mathbf{x}=\mathbf{b}$ has no solution. This in turn implies $\underset{\mathbf{B} \in \mathcal{B}_{TC}}{\max} \ |\mathbf{h}_R^H \mathbf{\Theta} \mathbf{h}_T|^2 < \bar{P}_R$. 
\end{proof}

\cref{theorem: tree-connected} reveals that the
\emph{sum} of the normalized Tx-Rx channel vectors $\hat{\mathbf{h}}_R+\hat{\mathbf{h}}_T$ determines a set of adversarial channels for the tridiagonal architecture. 
The phases of the entries of complex vector $\hat{\mathbf{h}}_R+\hat{\mathbf{h}}_T$ play an especially important role: when a pair of \emph{adjacent} entries of $\hat{\mathbf{h}}_R+\hat{\mathbf{h}}_T$ are \emph{real} scalar multiples of each other, 
the connection between the corresponding RIS elements is redundant, as the proof of \cref{theorem: tree-connected} shows. 
In this case, $\mathbf{B}$ can equivalently be written as a block-diagonal matrix corresponding to a group-connected RIS.
Hence, $\bar{P}_R$ cannot be achieved if the 
adversarial channel conditions in
\cref{proposition: disconnected} are satisfied. 
A full characterization of the set of adversarial channels (necessary and sufficient conditions) of the tridiagonal BD-RIS is part of ongoing work.

\section{Numerical results}
\label{sec: sims}

In this section, we numerically compare the received power at the Rx when the different considered BD-RIS architectures are deployed in a) Rayleigh fading channels and b) adversarial channels. We generate these vectors for RIS sizes $N = 8,16,...,64$, and conduct 1000 Monte Carlo trials for each $N$. We simulate the tridiagonal tree-connected (TC) architecture and group-connected (GC) architecture with equal group sizes 2 and 4, such that the number of groups is $G = N/2$ and $G = N/4$, respectively. The single-connected (conventional) RIS architecture $G = N$ is also included. 
The optimal values for $\mathbf{B}$ in case of the tree and group-connected (including single-connected) architectures are found using \cite[Algorithms 1 and 2]{BDRIS_graphtheory}, respectively. The average $P_R$ for each $N$ is plotted with respect to the upper bound in \eqref{eq: bound_full}. 

\cref{fig: GC_tree_random} shows the results for Rayleigh fading channels. The tree-connected model reaches \eqref{eq: bound_full}, which is consistent with \cite{BDRIS_graphtheory}, while the group-connected architecture reaches around 80\% ($G = N/2$) to 90\% ($G = N/4$) of this maximum value as $N$ grows. The single-connected architecture only reaches around 60\% of the maximum showing that while the performance of the group connected is not optimal, it is considerably better than that of the conventional single-connected RIS. 

We also generate channels that are favorable for the group-connected architecture, i.e., channels that do not satisfy the adversarial conditions in \cref{proposition: disconnected} for equal group size $2$. 
\cref{fig: GC_favorable} shows that, as expected, the tree-connected and group-connected architectures reach the optimum in this case.

\begin{figure}[ht]
  \begin{subfigure}[t]{0.45\columnwidth}
    \centering
    \includegraphics[width=\textwidth]{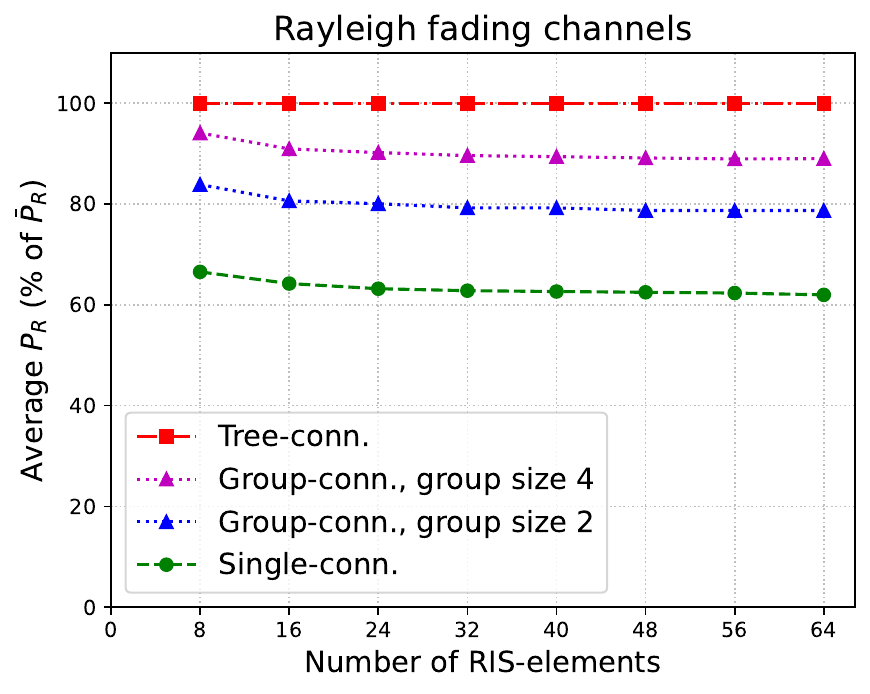}
    \caption {Rayleigh fading channels}
    \label{fig: GC_tree_random}
  \end{subfigure}
  \hfill
  \begin{subfigure}[t]{0.45\columnwidth}
    \centering
    \includegraphics[width=\textwidth]{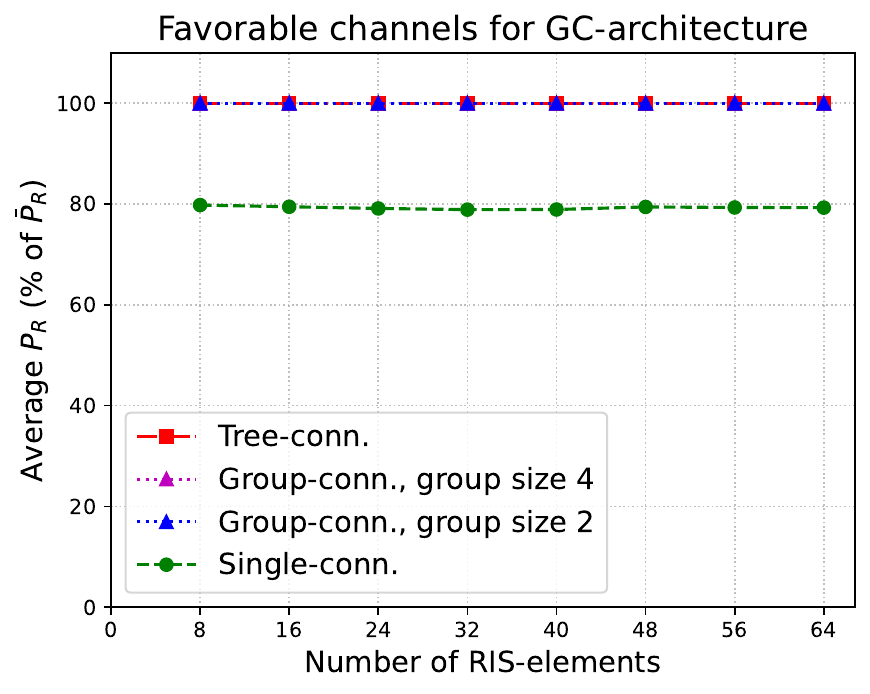}
    \caption{GC-favorable channels}
    \label{fig: GC_favorable}
  \end{subfigure}
  \caption{The TC architecture outperforms GC with Rayleigh fading channels while both achieve the upper bound when the channels are favorable for the GC architecture.}\vspace{-.5cm}
  \label{fig: random_and_favorable}
\end{figure}

Next, we demonstrate that adversarial channel realizations can significantly decrease the received signal power. First we generate $\mathbf{h}_R$ and $\mathbf{h}_T$ following the conditions in Proposition \ref{proposition: disconnected} for equal group size $4$. This is achieved by first independently drawing the real and imaginary parts of each entry of $\mathbf{h}_R$ from a uniform distribution. We then similarly draw $N/4$ random vectors $\{\mathbf{f}_g\}_{g=1}^{N/4}\subseteq \mathbb{C}^4$ and set 
$\mathbf{h}_T\!=\![a_1 \mathbf{f}_1^T/||\mathbf{f}_1||_2,a_2 \mathbf{f}_2^T/||\mathbf{f}_2||_2,\ldots,a_{N/4} \mathbf{f}_{N/4}^T/||\mathbf{f}_{N/4}||_2]^T$,
where $\{a_g\}_{g=1}^{N/4}$ are independent random variables uniformly distributed between 0 and 1. This ensures that vectors $(\mathbf{h}_R, \mathbf{h}_T)$ do not satisfy \eqref{eq: gc_reaches_upper_bound}. 
\cref{fig: GC_adversarial} shows that the group-connected architecture with group size 4 reaches around 70\% and with group size 2 only around 60\% of the maximum received power, which the tree-connected architecture achieves.

Finally, we generate adversarial channels for the tree-connected architecture using \cref{theorem: tree-connected}. We consider a Rayleigh fading Rx channel such that $\mathbf{h}_R$ is drawn from a zero mean complex Gaussian distribution. From the normalized $\hat{\mathbf{h}}_R$ we can generate $\hat{\mathbf{h}}_T$ by setting element $[\hat{\mathbf{h}}_T]_i = [\hat{\mathbf{h}}_R]_{i+1}$ and $[\hat{\mathbf{h}}_T]_{i+1} = [\hat{\mathbf{h}}_R]_i$ for 
$i = \bl{1,3,\ldots, Q < N}$. 
This results in channels $\mathbf{h}_R$ and $\mathbf{h}_T$ that satisfy the conditions for not reaching $\bar{P}_R$ in \eqref{eq: theorem}. 
\bl{An example of such a $(\mathbf{h}_R, \mathbf{h}_T)$-pair when $N = 2$ is $\mathbf{h}_R = [2j, 3+j]^T, \ \mathbf{h}_T = \tfrac{1}{2}[3+j, 2j]^T$.} 
\cref{fig: tree_adversarial} shows that, in this case, the tree-connected architecture achieves around 73\% of the maximum, while the group-connected architecture with group size 4 performs on the level of the conventional RIS. 
We emphasize that \emph{no} algorithm for the tree-connected architecture can achieve $P_R=\bar{P}_R$ in this adversarial case. Whether it is possible to significantly improve performance by employing another method than \cite[Algorithm 1]{BDRIS_graphtheory} for optimizing $\boldsymbol{\Theta}$ in such cases is a topic is left for future work. Interestingly, the channel vectors considered in the above example require generating either $\mathbf{h}_R$ or $\mathbf{h}_T$ adversarially, but not both (jointly). Hence, an adversary would need knowledge of one of the channel vectors (which can be arbitrary) and the capability to influence the other adversarially to disrupt the tree-connected BD-RIS.
\vspace{-.25cm}
\begin{figure}[ht]
  \begin{subfigure}[t]{0.45\columnwidth}
    \centering
    \includegraphics[width=\textwidth]{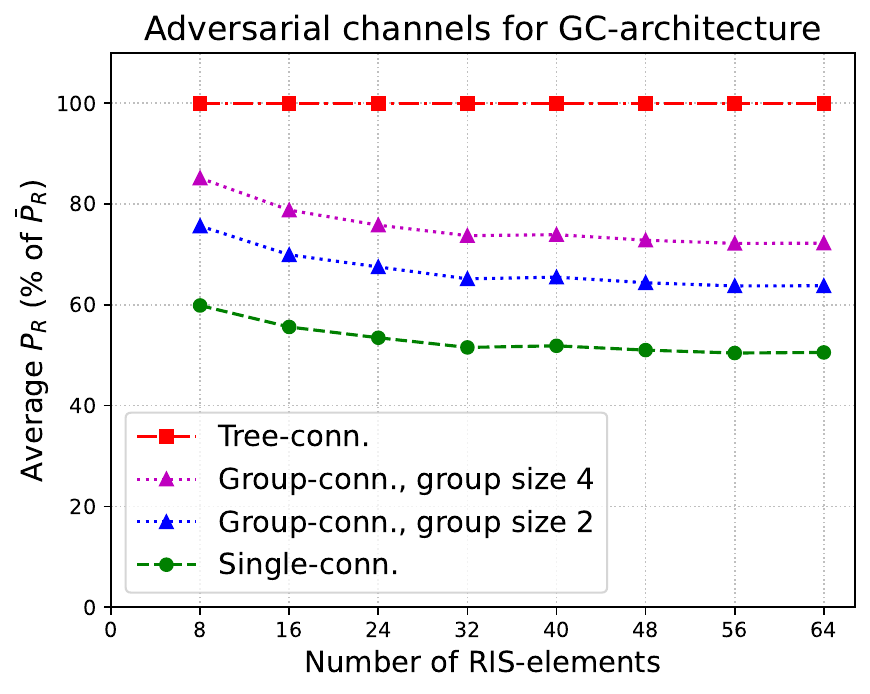}
    \caption{GC-adversarial channels}
    \label{fig: GC_adversarial}
  \end{subfigure}
  \hfill
  \begin{subfigure}[t]{0.45\columnwidth}
    \centering
    \includegraphics[width=\textwidth]{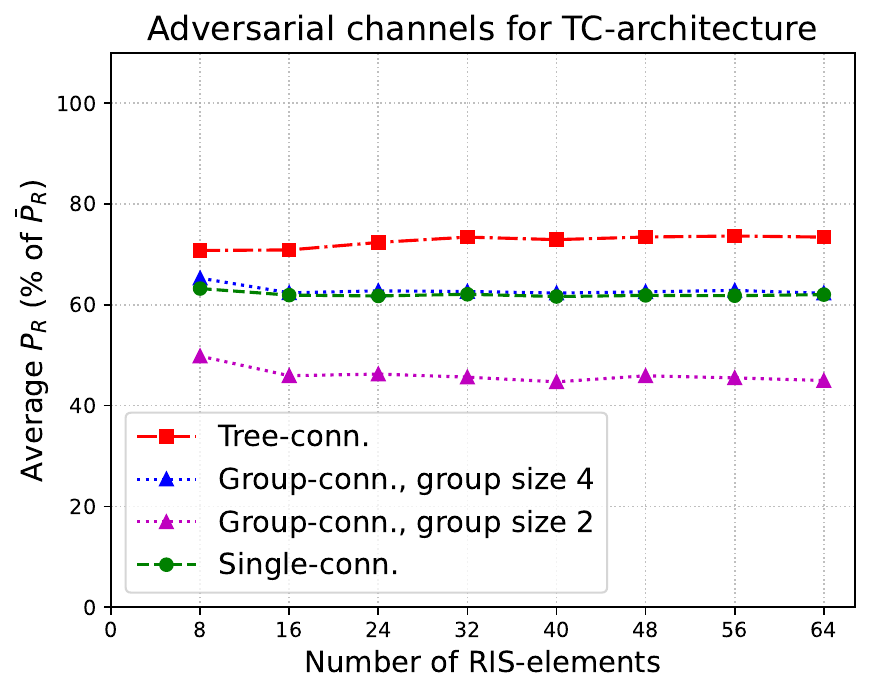}
    \caption{TC-adversarial channels}
    \label{fig: tree_adversarial}
  \end{subfigure}
  \caption{The TC architecture achieves the optimum received power with GC-adversarial channels. With TC-adversarial channels, the TC architecture achieves 70\% of the maximum.}\vspace{-.4cm}
  \label{fig: adversarial}
\end{figure}

\section{Conclusions}
\label{sec: conc}
\vspace{-.14cm}

This paper characterized adversarial channels in closed-form for both tree-connected (TC) \bl{(\cref{theorem: tree-connected})}
and group-connected 
\bl{(\cref{proposition: disconnected})} 
BD-RIS architectures. 
Simulations demonstrated a significant decrease in received signal power---up to 30\% in case of the TC architecture---in such adversarial scattering environments. 
The existence of adversarial channels is important to acknowledge, as they increase 
system vulnerability. 
While the probability of encountering such channels in nature may be low, there is potential for intentional or unintentional man-made interference. \bl{This includes} 
coarse quantization \bl{at the RIS, which leads} to a 
finite number of channel \bl{estimates}---a fact that a bad actor might exploit. 
\bl{Pertinent directions} for future work \bl{are} designing new low-complexity BD-RIS architectures robust to such adversarial propagation conditions, \bl{and investigating the MIMO case}.
%
\vspace{-.14cm}

\renewcommand\bibpreamble{\vspace{-0.13\baselineskip}}
\bibliography{references_new}

\end{document}